\documentclass[runningheads]{llncs}

\usepackage[T1]{fontenc}
\usepackage{graphicx}
\usepackage{float}
\usepackage{multirow}
\usepackage{lstautogobble}
\usepackage{stmaryrd}
\usepackage{amsmath}
\usepackage{cite}
\usepackage{fancyvrb}
\usepackage{bibnames}
\usepackage{graphicx}
\usepackage[pdftex,dvipsnames]{xcolor}  
\usepackage{amssymb}
\usepackage{listings}
\usepackage{multirow}
\usepackage{longtable}
\usepackage{arydshln}
\usepackage{caption}
\usepackage{subcaption}
\usepackage{enumitem}
\usepackage{amsfonts}
\usepackage{multicol}
\usepackage{lipsum}
\newif\iflong
\usepackage[noend]{algpseudocode}
\usepackage{xargs}
\usepackage{wrapfig}
\usepackage{lipsum}
\usepackage{floatflt}
\usepackage[linewidth=1.2pt,linecolor=red]{mdframed}
\usepackage{framed}
\usepackage{varwidth}
\usepackage{tikz}
\usepackage{ltl}
\usepackage{xspace}
\usepackage{todonotes}
\usepackage{longtable}
\usepackage[ruled,vlined,linesnumbered]{algorithm2e}
\usetikzlibrary{arrows,backgrounds,decorations,decorations.pathmorphing,positioning,fit,automata,shapes,snakes,patterns,plotmarks,calc,trees}

\usepackage{stackengine} 

\usepackage{hyperref}
\hypersetup{%
	colorlinks=true,           
	allcolors=blue!70!black,   
	pdfstartview=Fit,          
	breaklinks=true,
}
\usepackage{color}

\usepackage{graphicx}
\usepackage[firstpage]{draftwatermark}
\SetWatermarkText{\hspace*{10.75cm}\raisebox{19cm}{\includegraphics{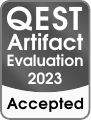}}}
\SetWatermarkAngle{0}
\begin{document}
	\newcommand{\plus}{{+}}
\newcommand{\code}[1]{\textsf{\small #1}\xspace}
\SetKwInOut{Input}{Input}
\SetKwInOut{Output}{Output}
\SetKwProg{Fn}{Function}{}{}
\DontPrintSemicolon

\newcommand{\quant}{\Q}

\newcommand{\setn}[1]{\ensuremath{\underline{#1}}}

\newcommand{\suchthat}{\,\,\ensuremath{|}\,\,}

\newcommand{\taskone}{1\xspace}
\newcommand{\tasktwo}{2\xspace}
\newcommand{\taskthree}{3\xspace}
\newcommand{\taskfour}{4\xspace}

\newcommand{\z}{\cellcolor{black}}
\newcommand{\x}{\cellcolor{lightgray}}

\renewcommand\UrlFont{\color{blue}\rmfamily}

\newcommand{\xpath}{\pi}
\newcommand{\Paths}[2]{\mathit{Paths}_{#1}^{#2}}
\newcommand{\fPaths}[2]{\mathit{fPaths}_{#1}^{#2}}

\newcommand{\dom}{\mathit{dom}}
\newcommand{\dtmc}{\mathcal{D}}
\newcommand{\dtmcs}{\mathbb{D}}
\newcommand{\hdtmcs}{\hat{\dtmcs}}
\newcommand{\emptydtmc}{\dtmc_{\emptyset}}

\newcommand{\hdtmc}{\hat{\dtmc}}
\newcommand{\PD}{\mathcal{P}}
\newcommand{\eval}{\mathit{eval}}

\newcommand{\mdp}{\mathcal{M}}
\newcommand{\mdps}{\mathbb{M}}
\newcommand{\hmdp}{\hat{\mdp}}
\newcommand{\hmdps}{\hat{\mdps}}
\newcommand{\modelsof}[1]{\textit{models}(#1)}

\newcommand{\HyperProb}{\textsf{\small HyperProb}\xspace}

\newcommand{\LTL}{\textsf{\small LTL}\xspace}
\newcommand{\CTL}{\textsf{\small CTL}\xspace}
\newcommand{\CTLstar}{\textsf{\small CTL$^*$}\xspace}
\newcommand{\PCTL}{\textsf{\small PCTL}\xspace}
\newcommand{\PCTLstar}{\textsf{\small PCTL$^*$}\xspace}
\newcommand{\HyperPCTL}{\textsf{\small HyperPCTL}\xspace}
\newcommand{\AHLTL}{\textsf{\small A-HLTL}\xspace}
\newcommand{\HyperPCTLS}{\textsf{\small HyperPCTL$^*$}\xspace}
\newcommand{\AHyperPCTL}{\textsf{\small AHyperPCTL}\xspace}
\newcommand{\ReachHyperPCTL}{\textsf{\small ReachHyperPCTL}\xspace}
\newcommand{\PHL}{\textsf{\small PHL}\xspace}
\newcommand{\NHyperPCTL}{\textsf{\small HyperPCTL}\xspace}
\newcommand{\HyperLTL}{\textsf{\small HyperLTL}\xspace}
\newcommand{\HyperCTL}{\textsf{\small HyperCTL}\xspace}
\newcommand{\HyperCTLstar}{\textsf{\small HyperCTL$^*$}\xspace}
\newcommand{\AFHyperLTL}{\mbox{AF-HyperLTL}\xspace}
\newcommand{\matching}{\mathcal{M}}
\newcommand{\topolgy}{\mathcal{T}}

\newcommand{\lang}{\mathcal{L}}
\newcommand{\Inf}{\mathsf{Inf}}
\newcommand{\pba}{\mathcal{P}}
\newcommand{\alphabet}{\mathrm{\Sigma}}
\newcommand{\state}{s}
\newcommand{\vstate}{\boldsymbol{\state}}
\newcommand{\states}{S}
\newcommand{\statespace}{\states}
\newcommand{\hstate}{\hat{\state}}
\newcommand{\hstateof}[2]{\hat{\state}_{#1}(#2)}
\newcommand{\hstateprimeof}[2]{\hat{\state}'_{#1}(#2)}
\newcommand{\hstatedoubleprimeof}[2]{\hat{\state}''_{#1}(#2)}

\newcommand{\bstate}{\bar{\state}}
\newcommand{\bstateof}[2]{\bar{\state}_{#1}(#2)}
\newcommand{\bstates}{\bar{\states}}

\newcommand{\map}{\mathsf{map}}

\newcommand{\scheduler}{\sigma}
\newcommand{\sched}{\scheduler}
\newcommand{\minscheduler}{\scheduler_{*}}
\newcommand{\maxscheduler}{\scheduler^{*}}
\newcommand{\vscheduler}{\boldsymbol{\scheduler}}
\newcommand{\schedulers}[1]{\Sigma^{#1}}

\newcommand{\hscheduler}{\hat{\scheduler}}
\newcommand{\hschedulerof}[1]{\hat{\scheduler}(#1)}
\newcommand{\hschedulers}{\hat{\schedulers{}}}

\newcommand{\bscheduler}{\bar{\scheduler}}
\newcommand{\bschedulers}{\hat{\schedulers{}}}

\newcommand{\btruth}[2]{\textit{isTrue}_{\,#1,#2}}
\newcommand{\ptruth}[2]{\textit{pr}_{\,#1,#2}}
\newcommand{\formula}{\varphi}
\newcommand{\ssqform}{\varphi^{sch}}
\newcommand{\sqform}{\varphi^{st}}
\newcommand{\nqform}{\varphi^{\textit{nq}}}
\newcommand{\arform}{\varphi^{\textit{ar}}}
\newcommand{\pathform}{\varphi^{\textit{path}}}

\newcommand{\propof}[2]{#1_{#2}}

\newcommand{\Trace}{\mathsf{Traces}}
\newcommand{\trace}{t}
\newcommand{\qtrace}{\eta}
\newcommand{\sform}{\mathrm{\Phi}}
\newcommand{\pform}{\varphi}

\newcommand{\action}{\alpha}
\newcommand{\altaction}{\beta}

\newcommand{\ins}{\textit{inst}}
\newcommand{\conc}{\circ}

\newcommand{\modes}{Q}
\newcommand{\mode}{q}
\newcommand{\modef}{\textit{mode}}

\renewcommand{\qed}{$~\blacksquare$}
\newcommand{\naturals}{\mathbb{N}_{>0}}
\newcommand{\naturalszero}{\mathbb{N}_{\geq 0}}

\newcommand{\F}{\LTLdiamond}
\newcommand{\G}{\LTLsquare}
\newcommand{\U}{\mathbin{\mathcal U}} 
\newcommand{\X}{\LTLcircle}
\newcommand{\Waitfor}{\,\mathcal W\,}
\newcommand{\suffix}[2]{#1[#2,\infty]}

\newcommand{\AP}{\mathsf{AP}}
\newcommand{\Next}{\X}
\newcommand{\Finally}{\F}
\newcommand{\Globally}{\G}
\newcommand{\V}{\mathcal{V}}

\newcommand{\pr}{\mathbb{P}}
\renewcommand{\Pr}{\mathrm{Pr}}
\newcommand{\prob}{P}
\newcommand{\Cyl}{\textit{Cyl}}
\newcommand{\parallelsum}{\mathbin{\|}}
\newcommand{\emptyword}{\epsilon}

\renewcommand{\topfraction}{0.96}
\renewcommand{\bottomfraction}{0.95}
\renewcommand{\textfraction}{0.1}
\renewcommand{\floatpagefraction}{1}
\renewcommand{\dbltopfraction}{.97}
\renewcommand{\dblfloatpagefraction}{.99}

\newcommand{\init}{\mathit{\iota_{init}}}
\newcommand{\tpm}{\mathbf{P}}
\renewcommand{\P}{\mathbf{P}}
\newcommand{\Act}{\mathit{Act}}
\newcommand{\act}{\mathit{act}}
\newcommand{\supp}{\mathit{supp}}
\newcommand{\start}{\mathit{init}}
\newcommand{\tru}{\mathtt{true}}
\newcommand{\fals}{\mathtt{false}}

\newcommand{\dbsim}{\mathit{dbSim}}
\newcommand{\res}{\mathit{res}}
\newcommand{\qout}{\mathit{qOut}}
\newcommand{\env}{\mathit{env}}
\newcommand{\fail}{\mathit{fail}}

\newcommand{\comp}[1]{\textsf{\small #1}}

\newcommand{\sigmakp}{$\mathsf{\Sigma^p_k}$\comp{-complete}\xspace}
\newcommand{\pikp}{$\mathsf{\Pi^p_k}$\comp{-complete}\xspace}

\newcommand\donotshow[1]{}

\definecolor{mGreen}{rgb}{0,0.6,0}
\definecolor{mGray}{rgb}{0.5,0.5,0.5}
\definecolor{mPurple}{rgb}{0.58,0,0.82}
\definecolor{backgroundColour}{rgb}{0.95,0.95,0.92}

\lstdefinestyle{CStyle}{
    backgroundcolor=\color{backgroundColour},   
    commentstyle=\color{mGreen},
    keywordstyle=\color{magenta},
    numberstyle=\tiny\color{mGray},
    stringstyle=\color{mPurple},
    basicstyle=\footnotesize,
    breakatwhitespace=false,         
    breaklines=true,                 
    captionpos=b,                    
    keepspaces=true,                 
    numbers=left,                    
    numbersep=2pt,                  
    showspaces=false,                
    showstringspaces=false,
    showtabs=false,                  
    tabsize=2,
    language=C
}

\newcommand{\borzoo}[1]{\textcolor{green}{#1}}

\newcommand{\Out}{\textit{Out}}
\newcommand{\In}{\textit{In}}
\newcommand{\symbprob}{\textit{Symb}\xspace}
\newcommand{\distance}{\textit{dist}}

\newcommand{\formH}{\textit{Form}_{\ReachHyperPCTL}}
\newcommand{\formR}{\textit{Form}_{\textit{NRA}}}
\newcommand{\exprH}{\textit{Expr}_{\ReachHyperPCTL}}
\newcommand{\exprR}{\textit{Expr}_{\textit{NRA}}}
\newcommand{\result}{\textit{color}}

\newcommand{\TA}{{\bf TA}\xspace}
\newcommand{\TS}{{\bf TS}\xspace}
\newcommand{\PW}{{\bf PW}\xspace}
\newcommand{\PC}{{\bf PC}\xspace}
\newcommand{\HS}{{\bf HS}\xspace}
\newcommand{\IJS}{{\bf IJ}\xspace}
\newcommand{\ROB}{{\bf RO}\xspace}

\newcommand{\TRUE}{{\bf T}\xspace}
\newcommand{\FALSE}{{\bf F}\xspace}
\newcommand{\WALLS}{{\bf W}\xspace}
\newcommand{\NWALLS}{{\bf NW}\xspace}
\newcommand{\GOAL}{{\bf G}\xspace}

\newcommand{\R}{\mathbb{R}}
\newcommand{\N}{\mathbb{N}}
\newcommand{\Z}{\mathbb{Z}}
\newcommand{\Q}{\mathbb{Q}}


\newcommand{\sep}{\ensuremath{~\mid~}}
\newcommand{\fal}{\mathtt{false}}
\renewcommand{\undef}{\bot}

\newcommand{\rew}{\mathcal{R}}
\newcommand{\re}{\textit{rew}}

\newcommand{\fullMDP}{(S, Act, \tpm, \AP, L)}
\newcommand{\fullMDPR}{(S, Act, \tpm, \AP, L, \re)}
\newcommand{\xor}{\ensuremath{\oplus}}
\newcommand{\impl}{\ensuremath{\rightarrow}}
\newcommand{\parComp}{|\!|}

\renewcommand{\phi}{\varphi}
\newcommand{\holdsToInt}{\textit{hInt}}
\newcommand{\holds}{\textit{h}}
\newcommandx{\holdsWith}[3][1=\statetup, 2=\phi, 3={\sschedtup}, usedefault]{\holds_{#1, #2}}
\newcommandx\holdsI[3][1=\statetup, 2=\phi, 3={\sschedtup}, usedefault]{\holdsToInt_{#1, #2}}
\newcommandx\probWith[3][1=\statetup, 2=\phi, 3={\sschedtup}, usedefault]{\textit{pr}_{#1, #2}}

\newcommand{\define}{\textit{def}}
\newcommand{\bsigma}{\textbf{$\sigma$}}
\newcommand{\fatr}{\vstate}

\renewcommand{\textfraction}{0.01}
\renewcommand{\floatpagefraction}{.99}
\renewcommand{\topfraction}{.99}
\renewcommand{\bottomfraction}{.99}

\newcommand{\od}[1]{\textcolor{blue}{#1}}
\newcommand*\samethanks[1][\value{footnote}]{\footnotemark[#1]}

\newcommand{\stutteract}{\varepsilon} 
\newcommand{\notstutact}{\overline\varepsilon}
\newcommand\ssched{\tau} 

\newcommand\Tau{\mathcal{T}}

\newcommand{\Mstutter}[1][\mdp]{#1 {\mathchoice
		{\raisebox{1pt}{$^{\stutteract}$}}
		{\raisebox{1pt}{$^{\stutteract}$}}
		{\raisebox{0.5pt}{$^{\stutteract}$}}
		{\raisebox{0.2pt}{$^{\stutteract}$}}}
}   

\newcommand\sSched[1][\mdp]{{\Tau}^{#1}}
\newcommandx{\Dstutter}[3][1=\mdp,2=\sched,3=\ssched,usedefault]{#1^{#2,#3}} 
\newcommand\refl [1] [\mdp]{#1}
\newcommand\DTMCtuple [1] [\refl]{(#1[S], \allowbreak \AP, \allowbreak #1[L], \allowbreak #1[\P] )}

\newcommandx{\ext}[3][1=\mdp,2=\sched,3=\ssched,usedefault]{#1^{#2,#3}} 

\newcommand\pquant{\varphi^{q}}
\newcommand\psched{\varphi^{sch}}
\newcommand\pstate{\varphi^{s}} 
\newcommand\pstutt{\varphi^{st}}
\newcommand\pnonquant{\varphi^{nq}}
\newcommand\pprob{\varphi^{pr}}
\newcommand\ppath{\varphi^{path}}

\newcommand\variable[1][s]{\hat{#1}}
\newcommand\varstate{\hat s}
\newcommand\varsched{\hat{\sched}}
\newcommand\varssched{\hat{\ssched}}

\newcommand\statetup{{\boldsymbol s}}
\newcommand\schedtup{{\boldsymbol \sched}}
\newcommand\sschedtup{{\boldsymbol \ssched}}
\newcommand\evaltups[1][\statetup]{\schedtup, #1, \sschedtup}
\newcommand{\numstutter}{n}
\newcommand{\numstates}{l}
\newcommandx\Probm [2][1=\mdp,2=\statetup,usedefault]{\textit{Pr}^{#1}_{#2}}

\renewcommand*{\iff}{\textit{ iff }}
\newcommand\squigm{
	\mathbin{\stackon[1pt]{$\rightsquigarrow$}{$\scriptscriptstyle \numstutter$}}
}

\renewcommand\implies{\Rightarrow}

\newcommand{\Splus}{\states_+}
\newcommand{\actiontup}{\boldsymbol{\alpha}}
\newcommand{\statetupmc}{\boldsymbol s}
\newcommand{\fsucc}{\mathit{succ}}
\newcommand{\go}{\textit{go}}
\newcommand{\htup}{\boldsymbol{h}}
\newcommand{\Tr}{\textit{tr}}

\newcommand{\funtup}{\boldsymbol{f}}
\newcommand{\stutterlength}{m}

\newcommand{\lIfElse}[3]{\lIf{#1}{#2 \textbf{else}~#3}}

\SetKwFunction{FMain}{Main}
\SetKwFunction{FSem}{Sem}
\SetKwFunction{FTruth}{Truth}
\SetKwFunction{FSemUUntil}{SemUnboundedUntil}
\SetKwFunction{FSemBUntil}{SemBoundedUntil}
\SetKwFunction{FSemUGlobally}{SemUnboundedGlobally}

\newcommand\phifull{
	\exists \varsched \phifromstate}
\newcommand\phifromstate{
	Q_1 \varstate_1(\varsched)\ldots Q_{\numstates} \varstate_{\numstates}(\varsched) \phifromstut }
\newcommand\phifromstut{
	\exists \varssched_1(\varstate_{k_1}) \ldots \exists \varssched_{\numstutter}(\varstate_{k_{\numstutter}}) . \pnonquant }

\newcommand{\CE}{{\bf CE}\xspace}
\newcommand{\TL}{{\bf TL}\xspace}
\newcommand{\ACDB}{{\bf ACDB}\xspace}

\newcommand{\orcid}[1]{\href{https://orcid.org/#1}{\includegraphics[width=3mm]{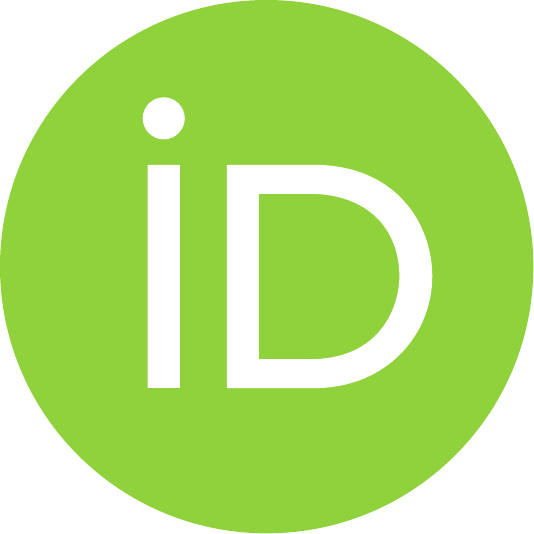}}}

\title{Introducing Asynchronicity to\\ Probabilistic Hyperproperties} 
\author{Lina Gerlach\inst{1}\orcid{0009-0002-5506-6181} \and
	Oyendrila Dobe\inst{2}\orcid{0000-0002-0799-1498} \and
	Erika {\'A}brah{\'a}m\inst{1}\orcid{0000-0002-5647-6134}\and 
	Ezio Bartocci\inst{3}\orcid{0000-0002-8004-6601}\and 
	Borzoo Bonakdarpour\inst{2}\orcid{0000-0003-1800-5419}}
\authorrunning{L. Gerlach et al.}
\institute{
	RWTH Aachen University, Aachen, Germany,\\
	\email{\{gerlach,abraham\}@cs.rwth-aachen.de}
	\and
	Michigan State University, East Lansing, MI, USA,\\ 
	\email{\{dobeoyen,borzoo\}@msu.edu} 
	\and
	Technische Universit\"at Wien, Vienna, Austria.\\
	\email{ezio.bartocci@tuwien.ac.at}}
\maketitle              
\begin{abstract} 

Probabilistic hyperproperties express probabilistic relations between different executions of systems with uncertain behavior. \HyperPCTL \cite{hyperpctl} allows to formalize such properties, where quantification over probabilistic schedulers resolves potential non-determinism. In this paper we propose an extension named \AHyperPCTL to additionally introduce \emph{asynchronicity} between the observed executions by quantifying over {\em stutter-schedulers}, which may randomly decide to delay scheduler decisions by idling.
To our knowledge, this is the first asynchronous extension of a probabilistic branching-time hyperlogic.%
We show that \AHyperPCTL can express interesting information-flow security policies, and propose a model checking algorithm for a decidable fragment.	%

\end{abstract}

	\section{Introduction}
\label{sec:intro}

\label{app:classic}
Consider the following simple multi-threaded program~\cite{smith03} consisting of two threads with a secret input $h$ 
and a public output $l$: 
\[ 
th \colon \text{\bf{while }} h>0 \text{\bf{ do }} \{h \leftarrow h-1\};\; l \leftarrow 2 
\quad \parallel \quad
th' \colon l \leftarrow 1  
\]
Assuming that this program is executed under a probabilistic scheduler, the probability of observing $l=1$ decreases for increasing initial values of $h$.
Hence, this program does not satisfy scheduler-specific probabilistic observational determinism (SSPOD)~\cite{minh2013confidentiality}, which requires that no information about the private data is leaked through the publicly visible data, for any scheduling of the threads.
In fact, the scheduler is creating a probabilistic side channel that leaks the value of the secret.
Probabilistic hyperlogics such as \HyperPCTL~\cite{hyperpctl,Dobe22a,atva20} and \PHL~\cite{dft20} are able to express and verify requirements such as SSPOD.

Interestingly, there is a way to mitigate this side channel similar to the padding mechanism that counters timing side channels. 
In the above example, for any two executions of the program under the same scheduler with different initial $h$, we can find \emph{stuttering variations} of the program such that the probability of reaching any specific final value of $l$ is the same for both executions.
For example, for two different values of $h$, say $h_1$ and $h_2$, where $h_1 <h_2$, letting thread $th'$ initially stutter $(h_2 - h_1)$ times (i.e., repeating the current state) in the execution starting with $h_1$ will equalize the probability of reaching $l=1$.

While there have been efforts to incorporate stuttering semantics in non-probabilistic logics (e.g., \AHLTL~\cite{bcbfs21}), in the probabilistic setting, neither \HyperPCTL nor \PHL allow reasoning about stuttering behaviors, i.e., their semantics are ``synchronous'' in the sense that all computation trees are evaluated in lock-step.
In this paper, we propose an asynchronous extension of~\HyperPCTL that allows to reason about stuttering computations and whether we can find stuttering variations of programs such that a probabilistic hyperproperty is satisfied.

\paragraph{Related Work}
\label{sec:related}

\HyperPCTL~\cite{hyperpctl} was the first logic for specifying probabilistic hyperproperties over DTMCs, by providing state-based quantifiers over the computation trees of the DTMC.
This logic was  further extended\cite{dabb21,Dobe22a,Dobe22reward,atva20,abbd20} with the possibility to specify 
quantifiers over schedulers for model checking Markov decision processes (MDPs). 
The probabilistic hyperlogic \PHL~\cite{dft20} can also handle analysis of MDPs. 
In general, the (exact) model checking problem
for both \HyperPCTL and \PHL is undecidable unless we restrict the class of schedulers or we rely on some
approximating methods. 
HyperPCTL*~\cite{0044NBP21,0044ZBP19} extends PCTL*~\cite{Baier2008} with quantifiers over execution paths and is employed in statistical model checking. 
All three logics are synchronous and lock-step. 
To the best of our knowledge, our work is the first to consider asynchronicity in the probabilistic setting. 

In the non-probabilistic setting, asynchronicity has already been studied \cite{BaumeisterCBFS21,Hsu2023,BozzelliPS21,BozzelliPS22,Beutner2022,BartocciFHNC22}.
In~\cite{BartocciFHNC22}, the authors study the expressivity of \HyperLTL~\cite{cfkmrs14} showing the impossibility  
to express the ``two-state local independence'' asynchronous hyperproperty, where information flow is allowed 
only after a change of state (for example in the case of declassification~\cite{SabelfeldS09}). To cope with this limitation, 
several asynchronous extensions of \HyperLTL have been proposed.

For example, \emph{Asynchronous HyperLTL}~\cite{Beutner2022} extends HyperLTL with quantification over ``trajectories'' that enable the alignment of execution traces from different runs. 
\emph{Stuttering HyperLTL}~\cite{BozzelliPS21} relates only stuttering traces where two consecutive 
observations are different. 
\emph{Context HyperLTL}~\cite{BozzelliPS21} instead allows to combine synchronous and 
asynchronous hyperproperties. All three logics are in general undecidable, but there are useful decidable fragments 
that can be model-checked. The expressiveness of these logics has been compared in \cite{BozzelliPS22}.

\paragraph{Contributions}
Our main contribution is a new logic, called \AHyperPCTL, which is an asynchronous extension of \HyperPCTL and allows to reason about probabilistic relations between stuttering variations of probabilistic and potentially non-deterministic systems.
To our knowledge, this is the first asynchronous extension of a probabilistic branching-time hyperlogic.
Our goal is to associate several executions with independent stuttering variations of the same program and compare them.
We implement this by extending \HyperPCTL with quantification over \emph{stutter-schedulers}, which specify when the program should stutter.

We show that \AHyperPCTL is useful to express whether information leaks can be avoided 
via suitable stuttering.
In the context of our introductory example, the following \AHyperPCTL formula expresses SSPOD under stuttering:
\begin{align*}
	\forall \varsched .\; 
	&\forall \variable(\varsched) .\; \forall \variable'(\varsched) .\;
	\exists \varssched(\variable) .\; \exists \varssched'(\variable') .\; \\
	&\left(h_{\varssched} {\neq} h_{\varssched'} \wedge \textit{init}_{\varssched} \wedge \textit{init}_{\varssched'}\right)
	\implies
	( \textstyle
	\bigwedge_{k \in \{1,2\}} \pr\big(\F(l{=}k)_{\varssched}\big) = \pr\big(\F(l{=}k)_{\varssched'}\big)
	),
\end{align*}
where $\varsched$ represents a probabilistic scheduler that specifies which thread is allowed to execute in which program state, $\variable$ and $\variable'$ represent initial states, and $\varssched$ and $\varssched'$ are stutter-scheduler variables for the computation trees rooted at $\variable$ and $\variable'$ under the scheduler $\varsched$.
This formula specifies that under any probabilistic scheduling $\varsched$ of the two threads, if we consider two computation trees 
starting in states $\variable$ and $\variable'$ with different values for the secret variable $h$, there should exist stutterings for the two experiments such that the probabilities of observing any specific final value of $l$ are the same for both.

We propose a model checking algorithm for \AHyperPCTL under restrictions on the classes of schedulers and stutter-schedulers.
Our method generates a logical encoding of the problem in real arithmetic and uses a satisfiability modulo theories (SMT) solver, namely Z3~\cite{dmb08}, to determine the truth of the input statement. 
We experimentally demonstrate that the model checking problem for asynchronous probabilistic hyperproperties is a computationally highly complex synthesis problem at two levels: both for synthesizing scheduler policies and for synthesizing stutter schedulers.
This poses serious problems for model checking: our current implementation does not scale beyond a few states. We discuss some insights about this complexity and suggest possible future directions.

\paragraph{Organization} 
We discuss preliminary concepts in Sec.~\ref{sec:prelim} and introduce \AHyperPCTL in Sec.~\ref{sec:logic}. We dedicate Sec.~\ref{sec:app} to applications and Sec.~\ref{sec:algo} to our algorithm. We discuss results of our prototype implementation in Sec.~\ref{sec:eval}. 
We conclude in Sec.~\ref{sec:conclusion} with a summary and future work.

	\section{Preliminaries}
\label{sec:prelim}

We denote the real (non-negative real) numbers by $\R$ ($\R_{\geq 0}$), and the natural numbers including (excluding) $0$ by $\N$ ($\N_{> 0}$). 
For any $n \in \N$, we define $[n]$ to be the set $\{0,\ldots, n{-}1\}$.
We use $()$ to denote the empty tuple and $\circ$ for concatenation.

\begin{definition}
\label{def:dtmc}
	A \emph{discrete-time Markov chain (DTMC)} is a tuple $\dtmc {=} (\states,\AP, L, \P)$ where (1) $\states$ is a non-empty finite set of \emph{states}, (2) $\AP$ is a set of \emph{atomic propositions}, (3) $L\colon \states \rightarrow 2^{AP}$ is a \emph{labeling function} and (4) $\P \colon \states \times \states \rightarrow [0,1]$ is a \emph{transition probability function} such that $\sum_{\state' \in \states} \P(\state, \state') =1$ for all $\state \in \states$.
\end{definition}
An \emph{(infinite) path} is a sequence $\state_0\state_1\state_2\ldots \in \states^\omega$ of states with $\P(\state_i, \state_{i+1}) > 0$ for all $i \geq 0$. Let $\Paths{\state}{\dtmc}$ denote the set of all paths of $\dtmc$ starting in $\state \in \states$, and $\fPaths{\state}{\dtmc}$ denote the set of all non-empty finite prefixes of paths from $\Paths{\state}{\dtmc}$, which we call \emph{finite paths}. For a finite path $\xpath =\state_0\ldots s_k \in \fPaths{\state_0}{\dtmc}$, $k \geq 0$, we define $|\xpath|=k$. A state $t \in \states$ is \emph{reachable} from $\state \in \states$ if there exists a finite path in $\fPaths{\state}{\dtmc}$ that ends in $t$.
The \emph{cylinder set} $\Cyl^{\dtmc}(\xpath)$ of a finite path $\xpath$ is the set of all infinite paths with $\xpath$ as a prefix. The
\emph{probability space} for $\dtmc$ and $\state\in\states$ is \\[1ex]
$
\hspace*{15ex} \big(\Paths{\state}{\dtmc},\big\{\bigcup_{\xpath\in R}\Cyl^{\dtmc}(\xpath)\suchthat
R\subseteq\fPaths{\state}{\dtmc}\big\},\Pr^{\dtmc}_{\state}\big)\ ,
$\\[1ex]
where the \emph{probability} of the cylinder set of $\xpath\in \fPaths{\state}{\dtmc}$ is
$\Pr^{\dtmc}_{\state}(\Cyl^{\dtmc}(\xpath))=\Pi_{i=1}^{|\xpath|}\P(\xpath_{i{-}1},\xpath_{i})$. These concepts have been discussed in detail in~\cite{Baier2008}.

\emph{Markov decision processes} extend DTMCs to allow the modeling of environment interaction or user input in the form of non-determinism.
\begin{definition}
\label{def:mdp}
	A \emph{Markov decision process (MDP)} is defined as a tuple $\mdp = (\states, \AP, L, \Act, \P)$, where 
	(1) $\states$ is a non-empty finite set of \emph{states}, 
	(2) $\AP$ is a set of \emph{atomic propositions}, 
	(3) $L \colon S \to 2^{\AP}$ is a \emph{labeling function}, 
	(4) $\Act$ is a non-empty finite set of \emph{actions}, 
	(5) $\P \colon \states \times \Act \times \states \rightarrow [0,1]$ is a \emph{transition probability function} such that for all $\state \in \states$ the set of its \emph{enabled actions}
	\[	\textstyle
		\Act(\state)= \big\{\action\in\Act \suchthat \sum_{\state' \in \states} \P(\state, 
		\action, \state') =1\big\}
	\]
	 is non-empty and 
	$\sum_{\state' \in \states} \P(\state, \alpha, \state') = 0$ for all $\alpha \in \Act \setminus \Act(s)$.
\end{definition}

Let $\mdps$ be the set of all MDPs.
For every execution step of an MDP, a \emph{scheduler} resolves the non-determinism by selecting an enabled action to be executed.

\begin{definition} 
	\label{def:scheduler}
	For an MDP $\mdp = (\states, \AP, L, \Act, \P)$, a \emph{scheduler} is a tuple $\scheduler = (\modes, \modef, \start, \act)$, where
	(1) $\modes$ is a non-empty countable set of \emph{modes}, 
	(2) $\modef \colon \modes \times \states \rightarrow \modes$ is a \emph{mode transition function}, 
	(3) $\start \colon \states \rightarrow \modes$ selects the \emph{starting mode} $\start(s)$ for each state of $s\in \states$, 
	and (4) $\act \colon \modes \times \states \times \Act \to [0,1]$ is a function with
	$\sum_{\alpha \in \Act(s)} \act(q,s,\alpha) = 1$ and $\sum_{\alpha \in \Act \setminus \Act(s)} \act(q,s,\alpha) = 0$
	for all $\state \in \states$ and $\mode \in \modes$.
\end{definition}

We use $\schedulers{\mdp}$ to denote the set of all schedulers for an MDP $\mdp$.
A scheduler is called \emph{finite-memory} if $\modes$ is finite, \emph{memoryless} if $\modes$ is a singleton, and 
\emph{deterministic} if $\act(\mode,\state,\alpha) \in \{0,1\}$ for all $(\mode,\state,\alpha) \in \modes \times \states \times \Act$.
If a scheduler is memoryless, we sometimes omit its only mode.

	\section{Asynchronous HyperPCTL}
\label{sec:logic}

\emph{Probabilistic hyperproperties} specify probabilistic relations between different executions of one or several probabilistic models. In previous work, we introduced \HyperPCTL~\cite{hyperpctl} to reason over non-determinism~\cite{Dobe22a} and rewards~\cite{Dobe22reward} for \emph{synchronous} executions, i.e., where all executions make their steps simultaneously. In this work, we propose an extension to reason about \emph{asynchronous} executions, where 
some of the executions may also \emph{stutter} (i.e., stay in the same state without observable changes) while others execute.

This is useful if, for example, the duration of some computations depend on some secret input and we thus might wish to make the respective duration unobservable.
A typical application area are multi-threaded programs, like the one presented in Section \ref{sec:intro}, where we want to relate the executions of the different threads. If there is a single processor available, each execution step allows one of the threads to execute, while the others idle. The executing thread, however, might also decide to stutter in order to hide its execution duration. To be able to formalize such behavior, the decision whether an execution stutters or not must depend not only on the history but also on the chosen action (in our example, corresponding to which of the threads may execute). 

In this section, we first introduce a novel scheduler concept that supports stuttering, followed by the extension of \HyperPCTL that we henceforth refer to as \AHyperPCTL. 
To improve readability, we assume that all executions run in the same MDP; an extension to different MDPs is a bit technical but straightforward.

\subsection{Stutter Schedulers}
\label{sec:stuttersched}

We define a \emph{stutter-scheduler} as an additional type of scheduler that only distinguishes between stuttering, represented by $\stutteract$, or proceeding, represented by $\notstutact$, for every state $\state \in \states$ and action $\alpha \in \Act$. 

\begin{definition} 
	\label{def:stutterscheduler}
	A \emph{stutter-scheduler} for an MDP $\mdp= (\states, \AP, L, \Act, \P)$ 
	is a tuple $\ssched = (\Mstutter[\modes], \Mstutter[\modef], \Mstutter[\start], \Mstutter[\act])$ where (1) $\Mstutter[\modes]$ is a non-empty countable set of \emph{modes}, (2) $\Mstutter[\modef] \colon \Mstutter[\modes] \times \states \times \Act \to \Mstutter[\modes]$ is a \emph{mode transition function}, (3) $\Mstutter[\start] \colon \states \to \Mstutter[\modes]$ is a function selecting a \emph{starting mode} $\start(s)$ for each state $s\in\states$ and (4) $\Mstutter[\act] \colon \Mstutter[\modes] \times \states \times \Act \times \{\stutteract, \notstutact\} \to [0,1]$ is a function with
	\[
	\Mstutter[\act](\Mstutter[\mode],s,\alpha, \stutteract) + \Mstutter[\act](\Mstutter[\mode],s,\alpha, \notstutact) = 1
	\]
	for all $(\Mstutter[\mode], \state, \alpha)  \in \Mstutter[\modes] \times \states \times \Act$.
\end{definition}		

We use $\sSched$ to denote the set of all stutter-schedulers for an MDP $\mdp$.
When reasoning about asynchronicity, we consider an MDP $\mdp$ in the context of a scheduler and a stutter-scheduler for $\mdp$.
At each state, first the scheduler chooses an action $\alpha$, followed by a decision of the stutter-scheduler whether to execute $\alpha$ or to stutter (i.e., stay in the current state). 
Thus, a stutter-scheduler makes its decisions based on not only its mode and the MDP state, but also depending on the action chosen by the scheduler.	

\begin{definition}
	\label{def:stutterinduceddtmc}
	For an MDP $\mdp= (\states, \AP, L, \Act, \P)$, a scheduler $\scheduler$ for $\mdp$ and a stutter-scheduler $\ssched$ for $\mdp$, the \emph{DTMC induced by $\scheduler$ and $\ssched$} is defined as $\Dstutter = \DTMCtuple[\Dstutter]$, where
	$\Dstutter[\states] = \modes \times \Mstutter[\modes] \times \states$,
	$\Dstutter[L](\mode,\Mstutter[\mode],\state) = L(\state)$ and 
	$\Dstutter[\P]((\mode,\Mstutter[\mode],\state), (\mode',\Mstutter[\mode]',\state')) =
	\begin{cases}
	  \textit{stut} & \textit{if } q'=q
          \neq \modef(q,s)
          \wedge \state'=\state \\
		\textit{cont} & \textit{if } \mode' = \modef(\mode,\state) \wedge (\mode' \neq \mode \vee \state' \neq \state)\\
		\textit{stut} + \textit{cont} & \textit{if } \mode'= \mode  = \modef(\mode,\state) \wedge \state'= \state
		\\
		0 & \textit{otherwise}  
	\end{cases}%
	$ 
	\begin{align*}
		&\textit{with } 
		\textit{stut} = \textstyle\sum_{\alpha \in \Act, \Mstutter[\modef](\Mstutter[\mode],\state,\alpha) = \Mstutter[\mode]'} \act(q,s, \alpha) \cdot \Mstutter[\act](\Mstutter[\mode],\state,\alpha, \stutteract) \\
		&\textit{and }\textit{cont} = \textstyle\sum_{\alpha \in \Act, \Mstutter[\modef](\Mstutter[\mode],\state,\alpha) = {\Mstutter[\mode]}'} \act(\mode,\state, \alpha) \cdot\Mstutter[\act](\Mstutter[\mode],\state,\alpha, \notstutact) \cdot \P(\state,\alpha, \state')\ .
	\end{align*}	
\end{definition}
The properties \emph{finite-memory}, \emph{memoryless}, and \emph{deterministic} for stutter schedulers are defined analogously as for schedulers.
A stutter-scheduler $\ssched$ for an MDP $\mdp$ is \emph{fair} for a scheduler $\sched \in \schedulers{\mdp}$ if the probability of taking infinitely many consecutive stuttering steps is 0. 
The different executions, whose relations we want to analyze, will be evaluated in the \emph{composition} of the induced models. 
The composition we use is the standard product of DTMCs with the only difference that we annotate atomic propositions with an index, indicating the execution in which they appear. 

\begin{definition}
  For $n\in\N$ and DTMCs $\dtmc_1, \ldots, \dtmc_n$ with $\dtmc_i=(\states_i, \AP_i, L_i, \P_i)$ for $i=1,\ldots,n$, we define the \emph{composition} $\dtmc_1\times\ldots\times\dtmc_n$ to be the DTMC $\dtmc=(\states,\AP, L, \P)$ with $\states=\states_1\times\ldots\times\states_n$, $\AP=\cup_{i=1}^{n}\{a_i\,|\,a\in\AP_i\}$, $L(s_1,\ldots,s_n)=\cup_{i=1}^n\{a_i\,|\,a\in L_i(s_i)\}$ and $P((s_1,\ldots,s_n),(s_1',\ldots,s_n'))=\Pi_{i=1}^n P_i(s_i,s_i')$.
\end{definition}

\begin{definition}
	For an MDP $\mdp$, $n\in\N_{> 0}$,
	a tuple $\schedtup = (\sched_1, \ldots, \sched_n) \in (\schedulers{\mdp})^n$ of schedulers, 
	and a tuple $\sschedtup = (\ssched_1, \ldots, \ssched_n) \in (\sSched[\mdp])^n$ of stutter-schedulers, 
	we define the \emph{induced DTMC} $\ext[\mdp][\schedtup][\sschedtup]=\ext[\mdp][\sched_1][\ssched_1] \times \ldots \times \ext[\mdp][\sched_n][\ssched_n]$.
\end{definition}

Later we will make use of \emph{counting} stutter-schedulers. These are deterministic bounded-memory stutter-schedulers which specify
for each state $\state\in\states$ and action $\alpha \in \Act(\state)$ a stuttering duration $j_{s, \alpha}$. Intuitively, $j_{s, \alpha}$ determines how many successive stutter-steps need to be made in state $s$ before $\alpha$ can be executed.

\begin{definition} 
	\label{def:continhstutterscheduler}
	An \emph{$m$-bounded counting} stutter-scheduler for an MDP $\mdp= (\states, \AP, L, \Act, \P)$ and $m\in\N_{> 0}$
	is a stutter-scheduler $\ssched = ([m], \Mstutter[\modef], \Mstutter[\start], \Mstutter[\act])$ such that for all $s\in\states$ and $\alpha\in\Act(s)$ there exists $j_{s,\alpha}\in[m]$ with (1) $\Mstutter[\start](s){=}0$ and (2) for each $j\in[m]$, if $j<j_{s,\alpha}$ then $\Mstutter[\modef](j,s,\alpha){=}j+1$ and $\Mstutter[\act](j,s,\alpha,\stutteract)=1$, and otherwise (if $j\geq j_{s,\alpha}$) $\Mstutter[\modef](j,s,\alpha)=0$ and $\Mstutter[\act](j,s,\alpha,\notstutact)=1$.
\end{definition}		

\begin{example}
	\label{ex:ssched}
	Consider the MDP $\mdp$ from Figure~\ref{fig:mc_illustr} as well as the DTMC $\ext$ induced by a probabilistic memoryless scheduler $\sched$ with $\sched(s_0, \alpha) = p$, $\sched(s_0, \beta) = 1-p$ for some $p \in [0,1]$ 
	and a 3-counting stutter-scheduler $\ssched$ on $\mdp$ with $j_{s_0, \alpha}=2$, $j_{s_0, \beta}=0$.
	The modes $q\in [3]$ of $\ssched$ store how many times we have stuttered since actually executing the last action. 
	For each state $(s,j)$ of $\ext$, first $\sched$ chooses an action $\alpha \in \Act(\state)$ probabilistically, and then $j$ is compared with the stuttering duration $j_{s,\alpha}$ stipulated by $\ssched$. 
	If we choose $\beta$ in state $(s_0, 0)$, then we move to a $\beta$-successor of $s_0$. 
	However, if we choose $\alpha$, then we move to the state $(s_0,1)$ and then choose an action again. If we choose $\beta$ at $(s_0,1)$, then we move to a $\beta$-successor of $s_0$, but if we choose $\alpha$ then we move to $(s_0,2)$.
	In particular, choosing $\alpha$ at $(s_0,0)$ does not mean that we  stutter twice in $s_0$. We stutter twice only if we also choose $\alpha$ in $(s_0,1)$.
	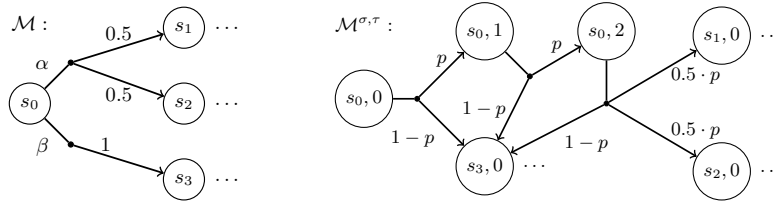
\begin{figure}[t]
	\centering
	\begin{subfigure}{0.32\textwidth}
		\centering 
		\scalebox{0.85}{
	\begin{tikzpicture}[->,line width=0.5pt,node distance=1.75cm]
		\tikzset{st/.style={draw,circle}} 
		\tikzset{ap/.style={font={\footnotesize},align=center}}
		\tikzstyle{dist}=[circle, inner sep=0.8pt, solid, draw=black,fill=black]
		\tikzset{trans/.style={font={\footnotesize},thick}}
		\node[st] (0) {$s_0$};
		\node[dist, below right=15pt of 0] (d0) {};
		\node[dist, above right=15pt of 0] (d1) {};
		%
		\node[st, right=of 0] (2) {$s_2$};
		\node[st, above=15pt of 2] (1) {$s_1$};
		\node[st, below=15pt of 2] (3) {$s_3$};
		\node[right=1pt of 2] (dots2) {\ldots};
		\node[right=1pt of 1] (dots1) {\ldots};
		\node[right=1pt of 3] (dots3) {\ldots};
		\node[above=20pt of 0] (title) {$\mdp:$};
		\path[trans]
		(0) edge [-] node[pos=0.5, above left] {$\alpha$} (d1)
		(d1) edge [] node[below] {$0.5$} (2)
		(d1) edge [] node[above] {$0.5$} (1)
		(0) edge [-] node[pos=0.5, below left] {$\beta$} (d0)
		(d0) edge [] node[above left] {$1$} (3);
	\end{tikzpicture}
		}
	\end{subfigure}
	\begin{subfigure}{0.60\textwidth}
		\centering
		\scalebox{0.8}{
	\begin{tikzpicture}[->,line width=0.5pt,node distance=1.75cm]
		\tikzset{st/.style={draw,circle}} 
		\tikzset{ap/.style={font={\footnotesize},align=center}}
		\tikzstyle{dist}=[circle, inner sep=0.8pt, solid, draw=black,fill=black]
		\tikzset{trans/.style={font={\footnotesize},thick}}
		\node[st] (00) {$s_0, 0$};
		\node[dist, right=10pt of 00] (d00) {};
		\node[st, above right=30pt of d00] (01) {$s_0, 1$};
		\node[dist, below right=15pt of 01] (d01) {};
		\node[st, below right=30pt of d00] (30) {$s_3, 0$};
		\node[st, right=30pt of 01] (02) {$s_0, 2$};
		\node[dist, below=19pt of 02] (d02) {};
		\node[right= of d02] (dd02) {}; 
		\node[st, below=15pt of dd02] (20) {$s_2, 0$};
		\node[st, above=15pt of dd02] (10) {$s_1, 0$};
		\node[right=1pt of 20] (dots2) {\ldots};
		\node[right=1pt of 10] (dots1) {\ldots};
		\node[right=1pt of 30] (dots3) {\ldots};
		\node[above=15pt of 00] (title) {$\ext:$};
		%
		\path[trans]
		(00) edge [-] node[pos=0.5, above] {} (d00)
		(d00) edge [] node[above] {$p$} (01)
		(d00) edge [] node[below left] {$1-p$} (30)
		(01) edge [-] node[pos=0.5, above] {} (d01)
		(d01) edge [] node[above] {$p$} (02)
		(d01) edge [] node[left] {$1-p$} (30)
		(02) edge [-] node[pos=0.5, above] {} (d02)
		(d02) edge [] node[right] {\;\;$0.5 \cdot p$} (10)
		(d02) edge [] node[right] {\;\;$0.5 \cdot p$} (20)
		(d02) edge [] node[below right] {$1-p$} (30);
	\end{tikzpicture}%
	}
	\end{subfigure}
	
	\caption{
		The interplay of a probabilistic memoryless scheduler $\sigma$ and a counting stutter-scheduler $\tau$ from Ex.~\ref{ex:ssched}. (The mode of $\sigma$ is omitted.)}
	\label{fig:mc_illustr}
\end{figure}

\end{example}

\subsection{Syntax}

To formalize relations of different executions, we begin an \AHyperPCTL formula as in \HyperPCTL \cite{Dobe22a} by first quantifying over the possible schedulers and then over the states of the MDP in which the respective execution under the chosen scheduler starts. 
In contrast to \cite{Dobe22a}, we now additionally quantify over stutter-schedulers in dependence on the chosen schedulers and initial states. 
Hence, the \emph{non-quantified} part of an \AHyperPCTL formula is evaluated on the DTMC(s) induced by not only the schedulers but also the stutter-schedulers, in accordance with Def.~\ref{def:stutterinduceddtmc}. 
Formally, we inductively define \AHyperPCTL scheduler-quantified formulas as follows:
\[
\begin{array}{lll}
	\mathit{\small scheduler-quantified \colon}
    & \psched &::= 
	\forall \variable[\sched] . \psched \mid 
	\exists \variable[\sched] . \psched \mid 
	\pstate 
	\\
	\mathit{\small state-quantified \colon}
	&\pstate &::= 
	\forall \hat{s}(\variable[\sched]) . \pstate \mid 
	\exists \hat{s}(\variable[\sched]) . \pstate \mid 
	\pstutt
	\\
	
	\mathit{\small stutter-quantified \colon}
	&\pstutt &::= 
	\forall \hat{\ssched}(\hat{s}) . \pstutt \mid 
	\exists \hat{\ssched}(\hat{s}) . \pstutt \mid 
	\pnonquant 
	\\
	
	\mathit{\small non-quantified \colon}
    &\pnonquant &::= 
	\texttt{true} \mid
	a_{\variable[\ssched]} \mid 
	\pnonquant \wedge \pnonquant \mid 
	\neg \pnonquant \mid 
	\pprob  \sim \pprob
	\\
	
	\mathit{\small probability\; expression \colon} &\pprob &::= 
	\pr(\ppath) \mid 
	f(\pprob_1, \ldots, \pprob_k) 
	\\
	
	\mathit{\small path\; formula\colon } &\ppath &::= 
	\Next \pnonquant \mid 
	\pnonquant \U \pnonquant 
\end{array}
\]
where $\varsched$ is a \emph{scheduler variable} from an infinite set $\variable[\Sigma]$, $\varstate$ is a \emph{state variable} from an infinite set $\variable[S]$, $\varssched$ is a \emph{stutter-scheduler variable} from an infinite set $\variable[\Tau]$,
$a \in \AP$ is an atomic proposition, $\sim \in \{\leq, <, =, \not=, >, \geq\}$, and
$f \colon [0,1]^k \to \R$ is a $k$-ary arithmetic operator over probabilities, where a constant $c$ is viewed as a $0$-ary function. 
$\pr$ refers to the probability operator and `$\Next$', `$\U$' refer to the temporal operators `next' and `until', respectively.

An \AHyperPCTL scheduler-quantified formula $\psched$ is \emph{well-formed} if each occurrence of any $a_{\varssched}$ for $\varssched \in \variable[\Tau]$ is in the scope of a {stutter quantifier} for $\varssched(\varstate)$ for some $\varstate \in \variable[S]$,  any quantifier for $\varssched(\varstate)$ is in the scope of a \emph{state quantifier} for $\varstate(\varsched)$ for some $\varsched \in \variable[\Sigma]$, 
and any quantifier for $\varstate(\varsched)$ is in the scope of a \emph{scheduler quantifier} for $\varsched$. \AHyperPCTL \emph{formulas} are well-formed \AHyperPCTL scheduler-quantified formulas, where we additionally allow standard syntactic sugar: $\fals = \neg 
\tru$, $\varphi_1 \vee \varphi_2 = \neg(\neg \varphi_1 \wedge \neg\varphi_2)$, $\varphi_1 \implies \varphi_2 = \neg( \varphi_1 \wedge \neg\varphi_2)$, $\F \varphi = \tru \U \varphi$, and $\pr(\G \varphi) = 1-\pr(\F \neg\varphi)$.

\newpage
\begin{example}
The  well-formed \AHyperPCTL formula
\[\exists \varsched .\; \forall \varstate(\varsched) .\; \forall \varstate'(\varsched) .\; \exists \varssched(\varstate) .\; \exists \varssched'(\varstate') .\; \big(\textit{init}_{\varssched} \wedge \textit{init}_{\varssched'}\big) \implies \big(\pr(\Finally a_{\varssched}) = \pr(\Finally a_{\varssched'})\big)\] 
states that there exists an assignment for $\varsched$, such that, if we start two independent \emph{experiments} from any state assignment to $\varstate$ and $\varstate'$, there exist independent possible stutter-schedulers for the two experiments, under which the probability of reaching a state labeled $a$ is equal, provided the initial states of the experiments are labeled $\textit{init}$. 
Further examples will be provided in Sec. \ref{sec:app}.
\end{example}

We restrict ourselves to quantifying first over schedulers, then over states, and finally over stutter-schedulers. 
This choice is a balance between the expressivity required in our applications and understandable syntax and semantics. 
Note that different state variables can share the same scheduler, but they cannot share the same stutter-scheduler.
Further, several different quantified stutter-schedulers in a formula are not allowed to depend on each other.

\subsection{Semantics}

The semantic judgement rules for \AHyperPCTL closely mirror the rules for \HyperPCTL \cite{atva20}.
\AHyperPCTL state formulas are evaluated in the context of an MDP $\mdp$, 
a sequence $\schedtup \in (\schedulers{\mdp})^{\numstutter}$ of schedulers, 
a sequence $\sschedtup \in (\sSched)^{\numstutter}$ of stutter-schedulers and 
a sequence $\statetup \in \ext[S][\schedtup][\sschedtup]$ of $\ext[][\schedtup][\sschedtup]$-states.
The length $\numstutter$ of these tuples corresponds to the number of stutter-schedulers in the given formula, which determines how many experiments run in parallel.
The elements of these tuples are instantiations of the corresponding variables in the formula.
We assume the stutter-schedulers to be fair and the variables used to refer to each of these quantifiers in the formula to be unique to avoid ambiguity. In the following, we use $\mathbb{Q}$ to refer to quantifiers $\{\forall, \exists\}$. We recursively evaluate the formula by instantiating the quantifier variables with concrete schedulers, states, and stutter-schedulers, and store them in sequences $\evaltups$. We begin by initializing each of these sequences as empty.
An MDP $\mdp=(\states, \AP, L, \Act, \P)$ satisfies an \AHyperPCTL formula $\phi$, denoted by 
$\mdp \models \varphi$, iff $\mdp, (), (), () \models \varphi$.

When instantiating a scheduler quantifier $\mathbb{Q} \variable[\sched] . \phi$ by a scheduler $\sigma$, we syntactically replace all occurrences of $\variable[\sched]$ in $\phi$ by $\sched$ and denote the result by $\varphi[\variable[\sched] \rightsquigarrow \sched]$\footnote{Note that we substitute a syntactic element with a semantic object in order to reduce notation; alternatively one could store respective mappings in the context.}.
The instantiation of a state quantifier $\mathbb{Q} \variable(\variable[\sched]) . \phi$ works similarly but it also remembers the respective scheduler: $\varphi[\hat{s} \rightsquigarrow s_{\sched}]$ denotes the result of syntactically replacing all occurrences of $\variable[\state]$ in $\phi$ by $s_{\sched}$. 
Finally, for instantiating the $\numstutter^{th}$ stutter-scheduler quantifier $\mathbb{Q} \variable[\ssched](\state_{\sched}) . \phi$, we replace
all occurrences of $a_{\variable[\ssched]}$ by $a_{\numstutter}$ and denote the result by $\varphi[\hat{\ssched} \squigm \ssched]$.
The semantics judgement rules for quantified and non-quantified state formulas, as well as probability expressions are defined as follows:
\[
\begin{array}{lcl}
	\mdp, \evaltups \models \forall \varsched .\ \varphi & \iff 
	& \forall \sched \in \schedulers{\mdp} .\ \mdp, \evaltups \models \varphi [\varsched \rightsquigarrow \sched]  \\
	\mdp, \evaltups \models \exists \varsched .\ \varphi & \iff 
	& \exists \sched \in \schedulers{\mdp} .\ \mdp, \evaltups \models \varphi [\varsched \rightsquigarrow \sched]  \\
	%
	\mdp, \evaltups \models \forall \varstate(\sched) .\ \varphi & \iff 
	& \forall s \in S .\ \mdp, \evaltups \models \varphi [\varstate \rightsquigarrow s_{\sched}]  \\
	\mdp, \evaltups \models \exists \varstate(\sched) .\ \varphi & \iff 
	& \exists s \in S .\ \mdp, \evaltups \models \varphi [\varstate \rightsquigarrow s_{\sched}] \\
	\mdp, \evaltups \models \forall \varssched(s_{\sched}) .\ \varphi & \iff 
	& \forall \ssched \in \sSched[\mdp] .\ \mdp, 
	\schedtup \circ \sched, 
	\statetup \circ (\start(s), \Mstutter[\start](s), s),\\
	&& \phantom{\forall  \ssched \in \sSched[\mdp] .\ }	\sschedtup \circ \ssched\models \varphi [\varssched \squigm \ssched] \\
	\mdp, \evaltups \models \exists \hat{\ssched}(s_{\sched}) .\ \varphi & \iff 
	& \exists \ssched \in \sSched[\mdp] .\ \mdp, 
	\schedtup \circ \sched, 
	\statetup \circ (\start(s), \Mstutter[\start](s), s), \\
	&& \phantom{\exists  \ssched \in \sSched[\mdp] .\ }\sschedtup \circ \ssched	\models \varphi [\varssched \squigm \ssched] \\
	%
	\mdp, \evaltups \models \texttt{true} \\
	\mdp, \evaltups \models a_i & \iff 
	& a_i \in \ext[L][\schedtup][\sschedtup](\statetup) \\
	\mdp, \evaltups \models \varphi_1 \wedge \varphi_2 & \iff 
	& \mdp, \evaltups \models \varphi_1 \text{ and } \mdp, \evaltups \models \varphi_2 \\
	\mdp, \evaltups \models \neg \varphi & \iff 
	& \mdp, \evaltups \not\models \varphi \\ 	
	\mdp, \evaltups \models \pprob_1 < \pprob_2 & \iff 
	& \llbracket \pprob_1 \rrbracket_{\mdp, \evaltups} <  \llbracket \pprob_2 \rrbracket_{\mdp,  \evaltups} \\
	%
	%
	\llbracket \pr(\ppath) \rrbracket_{\mdp, \evaltups} & = 
	& \Probm[{\ext[\mdp][\schedtup][\sschedtup]}]
	(\{ \pi \in \Paths{\statetup}{\ext[\mdp][\schedtup][\sschedtup]}
	\mid \mdp, \evaltups[\pi]  \models \ppath \}) \\
	\llbracket f(\pprob_1, \ldots, \pprob_k) \rrbracket_{\mdp, \evaltups} & = 
	& f(\llbracket \pprob_1 \rrbracket_{\mdp, \evaltups}\;, \ldots, \llbracket \pprob_k \rrbracket_{\mdp, \evaltups}) \\
\end{array}
\]
where
the tuples $\schedtup$, $\statetup$ and $\sschedtup$ are of length $\numstutter{-}1$. The semantics of path formulas is defined as follows for a path $\pi= \statetup_0 \statetup_1 \ldots $ of $\ext[][\schedtup][\sschedtup]$ with $\statetup_i \in \ext[S][\sched_1][\ssched_1] \times \ldots \times \ext[S][\sched_n][\ssched_n]$:
\[
\begin{array}{lll}
	\mdp, \evaltups[\pi] \models \Next \varphi & \iff 
	& \mdp, \evaltups[\statetup_1] \models \varphi 
	\\  
	\mdp, \evaltups[\pi] \models \varphi_1 \U \varphi_2 & \iff 
	& \exists j \geq 0 .\ \big( \mdp, \evaltups[\statetup_j] \models \varphi_2\; \wedge \\
	&& \phantom{\exists j \geq 0 .\ ( } \forall i \in [0,j) .\ \mdp, \evaltups[\statetup_i] \models \varphi_1 \big) \\ 
\end{array}
\]

\begin{lemma}
	\AHyperPCTL is strictly more expressive than \HyperPCTL.
\end{lemma}
\begin{proof}[Sketch]
For every MDP $\mdp$ and \HyperPCTL formula $\phi$, we can construct an MDP $\mdp'$ and an \AHyperPCTL formula $\phi'$ such that $\mdp \models_{\HyperPCTL} \phi$ iff $\mdp' \models_{\AHyperPCTL} \phi'$. 
The MDP $\mdp'$ is constructed from $\mdp$ by transforming each self-loop to a two-state-loop and then adding a unique label $a_s$ to each state $s$.
For this MDP, we define a formula $\mathit{trivial}_{\varssched_1, \ldots, \varssched_m}$ that checks whether the given stutter-schedulers are trivial by requiring that the probability of seeing the same state label $a_s$ in the current and in the next step must always be 0.
We construct $\phi'$ by adding a universal stutter-quantifier for each state quantifier and requiring that if these stutter-schedulers are all trivial, then the original non-quantified formula must hold.

\AHyperPCTL is thus at least as expressive as \HyperPCTL, and since \HyperPCTL cannot express stutter quantification, \AHyperPCTL is strictly more expressive.
\end{proof}

Hence, since the model checking problem for \HyperPCTL is already undecidable \cite{hyperpctl}, it follows that \AHyperPCTL model checking is undecidable as well.

\begin{theorem}
	The \AHyperPCTL model checking problem is undecidable.
\end{theorem}
	\newpage
\section{Applications of AHyperPCTL}
\label{sec:app}

\noindent\textbf{ACDB}\quad
\label{app:abcd}
Consider the code snippet \cite{guernicAutomatonbasedConfidentiality2007} in Fig.~\ref{fig:acbd}, where two threads synchronize across a critical region realized by the $\mathtt{await\ semaphore}$ command. Different interleavings of the threads can yield different sequences of observable outputs 
\begin{wrapfigure}[12]{r}{4cm}
		\vspace*{-5.5ex}
		\begin{lstlisting}[style=CStyle,tabsize=1,language=ML,basicstyle=\scriptsize,escapechar=/]
			/\color{magenta}Thread/ T1(){
				await semaphore{
					print(`a');
					v = v+1;
					print(`b');}
			}
			/\color{magenta}Thread/ T2(){
				print(`c');/\label{line:stutter_here}/
				if h=1{
					await semaphore{
						v = v+2;}}
				print(`d'); 
		}\end{lstlisting}
	\caption{\small Information leak.} 
	\label{fig:acbd}
\end{wrapfigure}
(i.e., permutations of $\mathtt{abcd}$).
Assume this program is executed according to a probabilistic scheduler.
Since the behavior of thread \texttt{T2} depends on a secret input \texttt{h},
under synchronous semantics, the program leaks information about the secret input: the probability of observing output sequence \texttt{acdb} is 0 if $\mathtt{h=0}$ and non-zero for $\mathtt{h=1}$.
However, stuttering after line~\ref{line:stutter_here} until \texttt{b} is printed would prevent an information leak. 
The following \AHyperPCTL formula expresses a variation of SSPOD, requiring that the probability of observing any specific output should be the same regardless of the value of \texttt{h}:
\begin{align*}
	\forall \varsched . \;&\forall \variable (\varsched).\; \forall \variable' (\varsched).\; \exists \varssched (\variable).\; \exists \varssched' (\variable').\  \\
	&\left(h_{\varssched} {\neq} h_{\varssched'} \wedge \textit{init}_{\varssched} \wedge \textit{init}_{\varssched'}\right) \textstyle
	\implies \left( \pr  \big(\square\bigwedge_{a \in \textit{obs}} \pr (\Next a_{\varssched}) = \pr (\Next a_{\varssched'})\big)= 1\right) 
\end{align*}

\noindent\textbf{Side-Channel Timing Leaks}\quad
\label{app:sidechannel}
are a kind of information leak where an attacker 
can infer the approximate 
value of the secret input based on the 
\begin{wrapfigure}[12]{r}{4.6cm}
	\vspace*{-6ex}
	\begin{lstlisting}[style=CStyle,tabsize=1,language=ML,basicstyle=\scriptsize,escapechar=/]
		void mexp(){
			c = 0; d = 1; i = k;
			while (i >= 0){
				i = i-1; c = c*2;
				d = (d*d) % n;
				if (b(i) = 1){
					c = c+1;
					d = (d*a) % n;}}
		}
		...
		t = new Thread(mexp()); 
		j = 0; m = 2 * k;
		while (j < m & !t.stop){
			j++;}\end{lstlisting}
	\caption{\small Modular exponentiation.}
	\label{fig:modexp}
\end{wrapfigure}
difference 
in execution time for different inputs to the algorithm. Stuttering could hide these differences.
Consider the code snippet in Fig.~\ref{fig:modexp} representing the modular exponentiation algorithm, which is part of the RSA public-key encryption protocol. It computes the value of $a^b \mod n$ where $\mathtt{a}$ (integer) is the plaintext and $\mathtt{b}$ (integer) is the encryption key. In~\cite{atva20}, we verified that we can notice the timing difference using a synchronous logic. We formalize in \AHyperPCTL that for \emph{any} possible scheduling of the two threads there exists possible stuttering that prevents the timing leak:
\begin{align*}
	\forall \varsched .\; 
	&\forall \variable(\varsched) .\; \forall \variable'(\varsched) .\;
	\exists \varssched(\variable) . \exists \varssched'(\variable') .\ \\
	&\left(h_{\varssched} {\neq} h_{\varssched'} \wedge \textit{init}_{\varssched} \wedge \textit{init}_{\varssched'}\right) 
	\implies
	\textstyle 
	\left(\bigwedge_{l=0}^{m} \pr(\Finally (j=l)_{\varssched}) = \pr(\Finally (j=l)_{\varssched'})\right)
\end{align*}%

	\section{Model Checking AHyperPCTL}
\label{sec:algo}

Due to general undecidability, we propose a model checking algorithm for a practically useful semantical fragment of \AHyperPCTL:
(1) we restrict scheduler quantification to probabilistic memoryless schedulers such that
	the same probabilistic decisions are made in states with identical enabled action sets, i.e., 
	if $\Act(\state) = \Act(\state')$, then $\act(\state, \alpha) = \act(\state', \alpha)$ for all $\alpha \in \Act(s)$, and
(2) stutter quantification ranges over $\stutterlength$-bounded counting stutter-schedulers.

These restrictions were chosen to achieve decidability but still be expressive enough for our applications.

For simplicity, here we describe the case for a single scheduler quantifier, but the algorithm can be extended to an arbitrary number of scheduler quantifiers. 
Additionally, we only describe the algorithm for existential scheduler and stutter quantification. 
The extension to purely universal
quantification is straightforward; we will discuss the handling of quantifier alternation
in Sec. \ref{sec:eval}.

\begin{figure}[t]
	\scalebox{0.95}{
		\begin{minipage}{1.04\linewidth}
			\begin{algorithm}[H]
			\caption{Main SMT encoding algorithm}
			\label{alg:main}
			\label{alg:truth}
			
			\KwIn{$\mdp=\fullMDP$: MDP; 
				$\stutterlength$: Memory size for the stutter-schedulers; 
				\\ \phantom{Input: }
				$\phifull$: 
				\AHyperPCTL formula.}
			\KwOut{Whether $\mdp$ satisfies the input formula.}
		\Fn{\FMain{$\mdp, \stutterlength, \phifull$}}{
		{$\varphi_{\textit{sch}}:= \bigwedge_{\emptyset{\not=} A \subseteq \Act} \left(\bigwedge_{\alpha \in A} 0\, {\leq}\, \sched_{A, \alpha}\, {\leq}\, 1\right) \wedge \sum_{\alpha \in A} \sched_{A, \alpha} {=} 1 $}\tcp*{scheduler choice} \label{line:actionenc}
		{$\phantom{\varphi_{\textit{sch}}:=}  \wedge \bigwedge_{i=1}^{\numstutter} \bigwedge_{s \in \states} \bigwedge_{\alpha \in \Act(s)} (\bigvee_{j=0}^{\stutterlength-1} \ssched_{i,s,\alpha} = j)$}\tcp*{stuttering choice}
		\label{line:stutterenc}
        {$\phantom{\varphi_{\textit{sch}}:=} \wedge
        	\bigwedge_{i=1}^{n} \bigwedge_{(\state,j)\in\states\times[m]}\bigwedge_{\alpha\in\Act(\state)}\bigwedge_{(\state',j')\in\fsucc(\state, \alpha)}  \varphi_{\go_{i, (s,j),\alpha,(s',j')}}$}\;
        \label{line:go}
        {$\phantom{\varphi_{\textit{sch}}:=} \wedge
        	\bigwedge_{i=1}^{n} \bigwedge_{(\state,j)\in\states\times[m]}\bigwedge_{\alpha\in\Act(\state)}\bigwedge_{(\state',j')\in\fsucc(\state, \alpha)}  \varphi_{\Tr_{i, (s,j),\alpha,(s',j')}}$}\;
        \label{line:Tr}
        \label{line:schedend}
        $ \varphi_{\textit{sem}} := \FSem(\mdp, \numstutter, \pnonquant)$ \tcp*{semantics of $\pnonquant$}\label{line:sem}
		\ForEach(\tcp*[f]{encode state quantifiers}){$i=1, \ldots, \numstates$}{
			\label{line:truthstart}
			\label{line:truthstatestart}
			\lIfElse{$Q_{i} {=} \forall$}
			{$A_i := $ ``$\bigwedge_{s_i \in S}$''}
			{$A_i := $ ``$\bigvee_{s_i \in S}$''}\label{line:truthstateend}
		}
		$ \varphi_{\textit{tru}} := A_1 \ldots A_{\numstates} \
			\Bigl(
			\holdsWith[((s_{k_{1}}, 0),\ldots,(s_{k_{\numstutter}}, 0))][\pnonquant]
			\Bigr)$\tcp*{truth of input formula}
		\label{line:truthend}
		\lIfElse{$check(\varphi_{\textit{sch}} \wedge  \varphi_{\textit{sem}}\wedge \varphi_{\textit{tru}} )=SAT$}{\Return \textit{TRUE}}{\Return \textit{FALSE}}
		}
	\end{algorithm}
	\end{minipage}
        }
\end{figure}

\begin{figure}[t]
	\scalebox{0.95}{
		\begin{minipage}{1.04\linewidth}
			\begin{algorithm}[H]
	\caption{SMT encoding for the meaning of the non-quantified formula}
	\label{alg:semEnc}
	
	\KwIn{$\mdp = \fullMDP$: MDP; $\numstutter$: number of experiments;
		\\ \phantom{Input: }
		$\phi$: quantifier-free \AHyperPCTL formula or expression.
	}
	\KwOut{SMT encoding of the meaning of $\phi$ for $\mdp$.}
	
	\Fn{\FSem{$\mdp, \numstutter, \phi$}}{
          	\If{$\phi$ is $\pr(\phi_1 \U \phi_2)$}{
			{$E := \FSem(\mdp, \phi_1, \numstutter) \wedge \FSem(\mdp, \phi_2, \numstutter)$\;}
			\ForEach{$\statetup=((s_1,j_1),\ldots,(s_n,j_n))\in (\states\times [m])^{\numstutter}$}{
				{$E := E \wedge (\holdsWith[][\phi_2] \implies \probWith{=}1)\wedge \bigl((\neg\holdsWith[][\phi_1] \wedge \neg\holdsWith[][\phi_2]) \implies \probWith{=}0 \bigr) $\;}
				{$E := E \wedge
						\Biggl[\Bigl[
						\holdsWith[][\phi_1] \wedge 
						\neg\holdsWith[][\phi_2] 
						\Bigr] 
						\implies 
						\biggl[
						\probWith = \displaystyle
							\sum_{\actiontup \in \Act(\statetup)}\;
							\sum_{\statetup' \in \fsucc(\statetup, \actiontup)}\;
							\big(
							\prod_{i=1}^n 
							\sched_{\Act(s_i),\alpha_i} \cdot
							\go_{i, \statetup_i,\alpha_i,\statetup_i'}
							\cdot 
							\Tr_{i, \statetup_i,\alpha_i,\statetup_i'} \big)
							\cdot \probWith[\statetup'][\phi]  
						\;\wedge$ $
						\Bigl[
						\probWith{>}0 \implies 
						\displaystyle
							\bigvee_{\actiontup \in \Act(\statetup)} 
							\bigvee_{\statetup' \in \fsucc(\statetup, \actiontup)}
						\bigl(\displaystyle 
						\prod_{i=1}^n 
						\sched_{\Act(s_i),\alpha_i}
						\cdot  
						\go_{i, \statetup_i,\alpha_i,\statetup_i'}
                                                {>}0
						\wedge 
						(\holdsWith[\statetup'][\phi_2] 
						\vee 
						d_{\statetup,\phi_2} {>} d_{\statetup',\phi_2} )\bigr)
						\Bigr]
						\biggr]\Biggr]$\;}
					\label{line:until}
			}	
		}
		\textbf{else if} 
		\ldots \\
		{\Return $E$\;}
	}
\end{algorithm}
\end{minipage}
}
\end{figure}

Our \AHyperPCTL model checking method adapts the \HyperPCTL algorithm \cite{atva20} with two major extensions: (1) we consider probabilistic memoryless schedulers instead of deterministic memoryless ones and (2) we support stuttering.
Assume in the following an MDP $\mdp=\fullMDP$, a memory bound $m$ for stutter-schedulers, and an input \AHyperPCTL formula $\varphi$. Our method generates a quantifier-free real-arithmetic formula $\varphi_{sch}\wedge \varphi_{tru}\wedge \varphi_{sem}$ that is satisfiable if and only if $\mdp\models\varphi$ (under the above restrictions on the domains of schedulers and stutter-schedulers).
The main method (Alg.~\ref{alg:main}) generates this encoding.
\\[0.5ex]
\ 1) In $\varphi_{\textit{sch}}$ (Lines~\ref{line:actionenc}--\ref{line:schedend}) we encode the scheduler probabilities and counting stutter-scheduler choices.
We use real-valued variables $\sched_{A, \alpha}$ to encode the probability of choosing $\alpha$ in state $s$ with $\Act(s) = A$, and variables $\ssched_{i, s, \alpha}$ with domain $[m]$ ($m$ being the stutter-scheduler memory bound) to represent the stuttering duration for state $s$ and action $\alpha$ under the $i$th stutter-scheduler quantifier.

For $\statetupmc=((s_1, j_1),\ldots,(s_n, j_n))\in(\states\times [m])^n$ we define $\Act(\statetupmc)=\Act(s_1)\times\ldots\times\Act(s_n)$.
The calculation of successor states for the encoding of the temporal operators depends on the chosen stutterings.
To describe possible successors, we use two mechanisms:
(i) For each $\statetup\in(\states\times [m])^n$ and $\boldsymbol{\alpha}\in\Act(\statetup)$ we define $\fsucc(\statetup, \boldsymbol{\alpha})$ to be the set of all $\statetup'=((s_1',j_1'),\ldots,(s_n',j_n'))\in (S \times [m])^n$ which under \emph{some} stutter-scheduler could be successors of $\statetup$ under $\actiontup$, i.e., such that 
\[\small
	\forall 1\leq i \leq n.\ 
	((j_i<m-1\wedge s_i=s_i'\wedge j_i'=j_i+1)\vee (\tpm(s_i, \alpha_i, s'_i)>0 \wedge j'_i = 0)) \, .
\]
(ii) 
For each $1\leq i\leq n$, $(s,j)\in (S\times [m])$, $\alpha \in \Act(s)$, and $(s',j') \in \fsucc(s,\alpha)$ we define a pseudo-Boolean variable $\go_{i, (s, j),\alpha,(s',j')}$ as well as a real variable $\Tr_{i, (s,j),\alpha,(s',j')}$ and define the formulas 
\begin{eqnarray*}
	\varphi_{\go_{i, (s, j),\alpha,(s',j')}} &=& (\go_{i, (s, j),\alpha,(s',j')}{=}0\vee\go_{i, (s, j),\alpha,(s',j')}{=}1)\wedge \big(\go_{i, (s, j),\alpha,(s',j')}{=}1
	\\
	&&\leftrightarrow ((j<\ssched_{i,s,\alpha}\wedge j'=j+1)\vee(j\geq\ssched_{i,s,\alpha}\wedge j'= 0)) \big)
	\\
	\varphi_{\Tr_{i, (s,j),\alpha,(s',j')}} &=& 
	(j'=j+1 \wedge \Tr_{i, (s,j),\alpha,(s',j')} =1) \vee \\
	&& (j'= 0 \wedge \Tr_{i, (s,j),\alpha,(s',j')} = \tpm(s,\alpha,s'))
	\ .
\end{eqnarray*}
2) In $\varphi_{\textit{sem}}$ (Line \ref{line:sem}) we encode the semantics of the quantifier-free part $\pnonquant$ of the input formula by calling Alg.~\ref{alg:semEnc}.
The truth of each Boolean subformula $\phi'$ of $\pnonquant$ at state sequence $\statetup$ is encoded in a Boolean variable $\holdsWith[][\phi']$. We also define variables $\holdsI[][\phi']$ for the integer encoding (i.e., 0 or 1) of $\holdsWith[][\phi']$, and real-valued variables $\probWith[\statetup][\phi'']$ for values of probability expressions $\phi''$.
\\[0.5ex]
\ 3) 
In $\varphi_{\textit{tru}}$ (Lines \ref{line:truthstart}--\ref{line:truthend}) we state the truth of the input formula by first encoding the state quantifiers (Lines \ref{line:truthstatestart}--\ref{line:truthstateend}) and then stating the truth of the quantifier-free part $\pnonquant$ under all necessary state quantifier instantiations $(s_1{,}{\ldots}{,}s_l)$, i.e., where experiment $i{\in}\{1{,}{\ldots}{,}n\}$ starts in state $s_{k_i}$ and stutter-scheduler mode $0$ (Line \ref{line:truthend}).

\smallskip

In Algorithm~\ref{alg:semEnc}, we recursively encode the meaning of atomic propositions and Boolean, temporal, arithmetic and probabilistic operators. 
Due to space constraints, here we present only the encoding for the temporal operator ``until'', and refer to the appendix for the full algorithm. 

For the encoding of the probability $\pr(\phi_1 \U \phi_2)$ that $\phi_1 \U \phi_2$ is satisfied along the executions starting in state $\statetup \in (\states\times [m])^{\numstutter}$, the interesting case, where $\phi_1$ holds in $\statetup$ but $\phi_2$ does not, is in Line \ref{line:until}. The probability is a sum over all possible actions and potential successor states. 
Each summand is a product over 
(i) the probability of choosing the given action tuple, 
(ii) pseudo-Boolean values which encode whether the potential successor state is indeed a successor state under the encoded stutter-schedulers, 
(iii) real variables encoding the probability of moving to the given successors under the encoded stutter-schedulers, and 
(iv) the probability to satisfy the until formula from the successor state.
We use real variables $d_{\statetup,\phi, \sschedtup}$ to assure that finally a $\phi_2$-state will be reached on all paths whose probabilities we accumulate.

Our SMT encoding is a Boolean combination of Boolean variables, 
non-linear real constraints, 
and linear integer constraints. 
Our linear integer constraints can be implemented as linear real constraints. 
The non-linear real constraints stem from the encoding of the probabilistic schedulers, not from the encoding of the stutter-schedulers.
We check whether there exists an assignment of the variables such that the encoding is satisfied. 
SMT solving for non-linear real arithmetic without quantifier alternation has been proven to be solvable in exponential time in the number of variables \cite{smtKroening,smtHandbook}. 
However, all available tools run in doubly exponential time in the number of variables \cite{cadCollins,mcsatMoura,coveringAbraham}. 
The number of variables of our encoding is exponential in the number of stutter quantifiers, and polynomial in the size of the formula, the number of states and actions of the model, and the memory size for the stutter-schedulers. 
Hence, in practice, our implementation is triple exponential in the size of the input. 
This yields an upper bound on the complexity of model checking the considered fragment.

The size of the created encoding is exponential in the number of stutter-schedulers and polynomial in the size of the \AHyperPCTL formula, the number of states and actions of the model and the memory-size for the stutter-schedulers.

	\section{Implementation and Evaluation}
\label{sec:eval}
We implemented the model checking algorithm described in Section~\ref{sec:algo} based on the existing implementation for \HyperPCTL, using the SMT-solver Z3 \cite{dmb08}. 
We performed experiments on a PC with a 3.60GHz i7 processor and 32GB RAM. 
Our implementation and case studies are available at \url{https://github.com/carolinager/A-HyperProb}.
It is important to note that checking the constructed SMT formula is more complicated than in the case for \HyperPCTL, since the SMT formula contains non-linear real constraints due to the probabilistic schedulers, whereas the SMT formula for \HyperPCTL contains only linear real arithmetic.

We optimized our implementation by reducing the number of variables as described in \cite{Dobe22a} based on the quantifiers relevant for the encoding of the considered subformula.
However, for interesting properties like the properties presented in Section~\ref{sec:app}, where we want to compare probabilities in different executions, we nevertheless have to create a variable encoding the validity of a subformula at a state for $(|\states|\cdot \stutterlength)^{\numstutter}$ combinations of states for each subformula.

All example applications presented in Sec.~\ref{sec:intro} and \ref{sec:app} consist of universal scheduler and state quantification, and existential stutter quantification. 
However, our implementation is restricted to existential scheduler and stutter quantification.
Allowing arbitrary quantifiers for scheduler and stutter quantification is also possible in theory, but we found it to be infeasible in practice.
For universal quantifiers we would need to check whether the SMT encoding holds for all possible assignments of the variables encoding the schedulers and stutter-schedulers, while for all other variables we only check for existence. 
This would yield an SMT instance with quantifier alternation, which is considerably more difficult than the current encoding with purely existential quantification \cite{cadBrown}. 
Alternatively, we can encode the semantics for each possible combination of stutter-schedulers separately, meaning that we have to use variables $h_{\statetup, \phi, \sschedtup}$ encoding the truth of $\phi$ at $\statetup$ under $\sschedtup$.
This makes the number of variables exponential in the number of stutter quantifiers and the number of states and actions of the model, and polynomial in the size of the formula and the memory-size for stuttering.
As a result, the implementation scales very badly, even after decreasing the number of variables via an optimization based on the relevant quantifiers. For all our applications, creating the semantic encoding exceeded memory after 30 minutes at the latest in this case.

Since neither option is a viable solution, for our case studies we instead consider the presented formulas with existential instead of universal scheduler quantification. 
In future work, it would be worth to explore how one could employ quantifier elimination to generate a set of possible schedulers from one scheduler instance satisfying the existential quantification.

\begin{table}[t]
	\centering
	\scalebox{0.83}{
		\begin{tabular}{|c|l||r|r|r||c|@{\,}c@{\,}|@{\,}c@{\,}||@{\,}c@{\,}|@{\,}c@{\,}|}
			\hline
			\multicolumn{2}{|c|}{\bf Case} & 
			\multicolumn{3}{c||}{\bf Running time ($s$)} & 
			\multicolumn{3}{c||}{SMT Solver} &
			\multicolumn{2}{c|}{Model}
			\\
			\cline{3-10} \multicolumn{2}{|c|}{\bf study} & {\bf Enc.} & {\bf Solving} & {\bf Total} & result & \#variables  & \#subform. & \#states & \#transitions \\
			\hline\hline
			
			\multirow{2}{*}{\CE} & $m=2$, $h=(0,1)$ & $0.25$ & $15.82$ & $16.07$ & sat & 829 & 341
			& 7 & 9 \\
			\cline{2-10}	
			& $m=2$, $h=(0,2)$  & $0.38$ & DNF & - & - & 1287 & 447 & 9 & 12 \\
			\hline\hline
			\TL & $m=2$, $k=1$ & $0.99$ & DNF & - & - & 3265 & 821 & 15
			& 23 \\
			\hline\hline 
			\ACDB & $m=2$ & $9.17$ & OOM & - & - & 14716 & 1284 & 24 & 36 \\
			\hline
	\end{tabular}}
	{%
		\vspace{2mm}
		\caption{Experimental results. \CE: Classic example, \TL: Timing leakage, \ACDB: Output information leak. DNF: did not finish, OOM: out of memory.}
		\label{tab:results}
	}
\vspace{-1.5em}
\end{table}

The results of our case studies are presented in Table~\ref{tab:results}.
Our first case study, \CE, is the classic example presented in Sec.~\ref{app:classic}. 
We compare executions with different initial values $h_1$ and $h_2$, denoted by $h=(h_1,h_2)$.
For $h=(0,1)$, the property can already be satisfied for the smallest non-trivial memory size $m=2$. For higher values of $h_2$, however, the SMT solver does not finish solving after 1 hour.
The second case study, \TL, is the side-channel timing leak described in Sec.~\ref{app:sidechannel}. We found that already for encryption key length $k=1$, and memory size $m=2$, the SMT solver did not finish after 1 hour even for a smaller formula, where we restrict the conjunction to the case $l=0$. 
Our third case study, \ACDB, is the output information leak presented in Sec.~\ref{app:abcd}. Here, the SMT solving exceeds memory after 18 minutes, even if we check only part of the conjunction.

We see several possibilities to improve the scalability of our implementation.
Firstly, we could experiment with different SMT solvers, like cvc5~\cite{cvc5}.
Secondly, we could parallelize the construction of encodings for different stutter-schedulers, and possibly re-use sub-encodings that are the same for multiple stutter-schedulers.
Another possibility would be to turn towards less accurate methods, and employ Monte Carlo or statistical model checking approaches.


	\section{Conclusions}
\label{sec:conclusion}

We proposed a new logic, called \AHyperPCTL, which is, to our knowledge, the first asynchronous probabilistic hyperlogic.
\AHyperPCTL extends \HyperPCTL by quantification over stutter-schedulers, which allow to specify when the execution of a program should stutter.
This allows to verify whether there exist stuttering variants of programs that would prevent information leaks, by comparing executions of different stuttering variations of a program. 
\AHyperPCTL subsumes \HyperPCTL. Therefore, the \AHyperPCTL model checking problem on MDPs is, in general, undecidable.
However, we showed that the model checking is decidable if we restrict the quantification to probabilistic memoryless schedulers and deterministic stutter-schedulers with bounded memory.
Since our prototype implementation does not scale well, future work could investigate the use of other SMT solvers, statistical model checking, or Monte Carlo methods, as well as a feasible extension to quantifier alternation for scheduler and stutter quantifiers.

\subsubsection{Acknowledgements} Lina Gerlach is supported by the DFG RTG 2236/2 \textit{UnRAVeL}. Ezio Bartocci is supported by the Vienna Science and Technology Fund (WWTF) [10.47379/ICT19018]. This work is also partially sponsored by the United States NSF SaTC awards 2245114 and 2100989.

	\bibliographystyle{splncs04}
	\bibliography{bibliography}
	
	\appendix
	\section{Appendix}
\label{sec:appendix}

\begin{figure}[h!tb]
	\scalebox{0.95}{
		\begin{minipage}{1.04\linewidth}
			\begin{algorithm}[H]
	\caption{Full SMT encoding for the meaning of the input formula}
	\label{alg:sem_full}
	
	\KwIn{$\mdp = \fullMDP$: MDP; $\stutterlength$: number of experiments;
		\\ \phantom{Input: }
		$\phi$: quantifier-free \AHyperPCTL formula or expression.
	}
	\KwOut{SMT encoding of the meaning of $\phi$ in $\mdp$.}
	
	\Fn{\FSem{$\mdp, \stutterlength, \phi$}}{
		\lIf{$\phi$ is $\mathtt{true}$}
		{$E := \bigwedge_{\statetup \in \Splus^{\numstutter}} \holdsWith$\tcp*[f]{$\Splus = S \times [m]$}}
		\ElseIf{$\phi$ is $a_{\variable[\ssched]_i}$}{
			$E := (\bigwedge_{\statetup \in \Splus^{\numstutter}, a \in L(s_{i})} \holdsWith) \wedge (\bigwedge_{\statetup \in \Splus^{\numstutter}, a \not\in L(s_{i})} \neg\holdsWith) $\;
		}
		\ElseIf{$\phi$ is $\phi_1 \wedge \phi_2$}{
			{$E : = \FSem(\mdp, \stutterlength, \phi_1) \wedge \FSem(\mdp, \numstutter, \phi_2)$\;}
			{$E := E \wedge 
				\bigwedge_{\statetup \in \Splus^{\numstutter}} ((\holdsWith \wedge \holdsWith[][\phi_1] \wedge \holdsWith[][\phi_2]) \vee (\neg\holdsWith \wedge \neg(\holdsWith[][\phi_1] \wedge \holdsWith[][\phi_2])))$\;}
		}
		\ElseIf{$\phi$ is $\neg\phi'$}{
			$E : = \FSem(\mdp, \stutterlength, \phi') \wedge \bigwedge_{\statetup \in \Splus^\numstutter} (\holdsWith \oplus \holdsWith[][\phi'])$\;
		}
		\ElseIf{$\phi$ is $\phi_1 < \phi_2$}{
			{$E := \FSem(\mdp, \stutterlength, \phi_1) \wedge \FSem(\mdp, \stutterlength, \phi_2)$ \;}
			\ForEach{$\statetup = ((s_1, j_1), \ldots, (s_{\numstutter}, j_{\numstutter})) \in \Splus^{\numstutter}$ }{
				{$E := E \wedge$ $
				\bigl( (\holdsWith \wedge \probWith[][\phi_1] < \probWith[][\phi_2]) \vee  (\neg\holdsWith \wedge \probWith[][\phi_1] \geq \probWith[][\phi_2]) \bigr) $\;}
			}
		}
		\ElseIf{$\phi$ is $\pr(\Next \phi')$}{
			$E := \FSem(\mdp, \stutterlength, \phi')$\; 
			\ForEach{$\statetup = ((s_1, j_1), \ldots, (s_{\numstutter}, j_{\numstutter})) \in \Splus^{\numstutter}$}{
				{$E := E \wedge
					\Bigl((\holdsI[][\phi'] {=} 1 \wedge \holdsWith[][\phi']) \vee (\holdsI[][\phi'] {=} 0 \wedge \neg\holdsWith[][\phi'])\Bigr)$\;}
				{$E := E \wedge \displaystyle 
					\probWith = \displaystyle
					\sum_{\actiontup \in \Act^n}
					\sum_{\statetup' \in \fsucc(\statetup, \actiontup)}
					\;\prod_{i=1}^n 
					(\go_{i, \statetup_i, \alpha_i, \statetup'_i} \cdot
					\sched_{\Act(s_i),\alpha_i} \cdot
					\Tr_{i, \statetup_i,\alpha_i,\statetup_i'})
					\cdot
					\holdsI[\statetup'][\phi']
					$\;}
			}
		}
		\lElseIf{$\phi$ is $\pr(\phi_1 \U \phi_2)$}{$E:= \FSemUUntil(\mdp, \stutterlength, \phi)$}
		\lElseIf{$\phi$ is $c$}{$E := \bigwedge_{\statetup \in \Splus^{\numstutter}} (\probWith = c)$} 
		\ElseIf(\tcp*[f]{$\mathop{op} \in \{+, -, *\}$}){$\phi$ is $\phi_1 \mathrel{op} \phi_2$}{
			{$E := \FSem(\mdp, \stutterlength, \phi_1) \wedge \FSem(\mdp, \stutterlength, \phi_2)$\;}
			{$E := E \wedge 
				\bigwedge_{\statetup \in \Splus^{\numstutter}} (\probWith\ {=}\ (\probWith[][\phi_1] \mathop{op} \probWith[][\phi_2]))$\;}
		}
		{\Return $E$\;}
	}
\end{algorithm}
\end{minipage}
}
\end{figure}

\begin{figure}[H]
	\scalebox{0.95}{
		\begin{minipage}{1.04\linewidth}
			\begin{algorithm}[H]
	\caption{SMT encoding for the meaning of unbounded until formulas}
	\label{alg:semUuntil_full}
	
	\KwIn{$\mdp = \fullMDP$: MDP; $\stutterlength$: number of experiments;
		\\ \phantom{Input: }
		$\phi$: \AHyperPCTL unbounded until formula.
	} 
	\KwOut{SMT encoding of the meaning of $\phi$ in $\mdp$.}
	
	\Fn{\FSemUUntil{$\mdp, \stutterlength, \phi = \pr(\phi_1 \U \phi_2)$}}{
	{$E := \FSem(\mdp, \phi_1, \numstutter) \wedge \FSem(\mdp, \phi_2, \numstutter)$\;}
	\ForEach{$\statetup=((s_1,j_1),\ldots,(s_n,j_n))\in (\states\times [m])^{\numstutter}$}{
		{$E := E \wedge (\holdsWith[][\phi_2] \implies \probWith{=}1)\wedge \bigl((\neg\holdsWith[][\phi_1] \wedge \neg\holdsWith[][\phi_2]) \implies \probWith{=}0 \bigr) $\;}
			{$E := E \wedge
				\Biggl[\Bigl[
				\holdsWith[][\phi_1] \wedge 
				\neg\holdsWith[][\phi_2] 
				\Bigr] 
				\implies 
				\biggl[
				\probWith = \displaystyle
				\sum_{\actiontup \in \Act(\statetup)}\;
				\sum_{\statetup' \in \fsucc(\statetup, \actiontup)}\;
				\big(
				\prod_{i=1}^n 
				\sched_{\Act(s_i),\alpha_i} \cdot
				\go_{i, \statetup_i,\alpha_i,\statetup_i'}
				\cdot 
				\Tr_{i, \statetup_i,\alpha_i,\statetup_i'} \big)
				\cdot \probWith[\statetup'][\phi]  
				\;\wedge$ $
				\Bigl[
				\probWith{>}0 \implies 
				\displaystyle
				\bigvee_{\actiontup \in \Act(\statetup)} 
				\bigvee_{\statetup' \in \fsucc(\statetup, \actiontup)}
				\bigl(\displaystyle 
				\prod_{i=1}^n 
				\sched_{\Act(s_i),\alpha_i} \cdot
				\go_{i, \statetup_i,\alpha_i,\statetup_i'}
				{>}0   
				\wedge 
				(\holdsWith[\statetup'][\phi_2] 
				\vee 
				d_{\statetup,\phi_2} {>} d_{\statetup',\phi_2} )\bigr)
				\Bigr]
				\biggr]\Biggr]$\;}
	}
	{\Return $E$\;}	
	}
\end{algorithm}
\end{minipage}
}
\end{figure}

\end{document}